\documentclass[smallextended]{svjour3}       
\smartqed  
\usepackage{latexsym}
\usepackage[dvips]{graphics}
\usepackage{amssymb, amsmath}
\usepackage{bm}

\usepackage{comment}
\usepackage[dvips]{graphicx}
\usepackage{amssymb, amsmath}

\usepackage{color}

\def\R{\mathbb{R}}
\def\E{\mathbb{E}}
\def\d{\hbox{d}}

\begin{document}

\title{On traveling wave solutions to Hamilton-Jacobi-Bellman equation with inequality constraints\thanks{The first author (NI) is partially supported by Grant-in-Aid for Scientific Research (C) No. 21540117 from Japan Society for the Promotion of Science (JSPS), the second author (DS) is grateful to support  of VEGA 1/0747/12 grant and great hospitality during his visit of Hitotsubashi University in Tokyo.
}
}

\titlerunning{Traveling wave solutions to a Hamilton--Jacobi--Bellman equation}

\author{Naoyuki~ISHIMURA \and Daniel~\v{S}EV{\v C}OVI{\v C}}

\institute{
N. Ishimura \at Graduate School of Economics, 
Hitotsubashi University, Kunitachi, Tokyo 186-8601, Japan.
\email{ishimura@econ.hit-u.ac.jp}
\and
D. \v{S}ev\v{c}ovi\v{c} \at
Department of Applied Mathematics 
and Statistics, Faculty of Mathematics, Physics and Informatics, Comenius 
University, 842 48 Bratislava, Slovakia.
\email{sevcovic@fmph.uniba.sk}
}

\date{Received: date / Accepted: date}

\maketitle

\begin{abstract} 
The aim of this paper is to construct and analyze solutions to a class of Hamilton--Jacobi--Bellman equations with range bounds on the optimal response variable. Using the Riccati transformation we derive and analyze a fully nonlinear parabolic partial differential equation for the  optimal response function. We construct monotone traveling wave solutions and identify parametric regions for which the traveling wave solution has a positive or negative wave speed. 

\keywords{Hamilton-Jacobi-Bellman equation \and traveling wave solution \and
Riccati transformation \and stochastic dynamic programming }
\subclass{35K55 \and 34E05 \and 70H20 \and 91B70 \and 90C15 \and 91B16}
\end{abstract}

\maketitle

\section{Introduction}
\label{sec:intro}

\noindent
The purpose of this paper is to analyze special solutions to a fully nonlinear  partial differential equation which can be derived from the  Hamilton--Jacobi--Bellman (HJB) equation for the value function arising in a class of optimal allocation problems. In many practical stochastic dynamic optimization problems, the goal is to maximize the expected value of the terminal utility. More precisely, let us suppose that 
$X=X_t^\theta, t\in[0,T],$ is a stochastic process satisfying a stochastic 
differential equation (SDE):
\begin{equation}
\d X_t^\theta = \mu_t^\theta(X_t^\theta) \d t + \sigma_t^\theta(X_t^\theta) \d W_t,
\label{genproc}
\end{equation}
where $\mu_t^\theta$ and $\sigma_t^\theta>0$ are the drift and volatility of  
It\=o's stochastic process (\ref{genproc}). Here $W_t$ $(t\ge0)$ denotes the 
standard Wiener process. The goal is to find an optimal response strategy 
$\{\theta\}=\{\theta_t \,|\, t\in[0,T]\}$ belonging to a set ${\mathcal A}$ of 
admissible strategies and yielding the maximal expected utility from the terminal
value $X_T^\theta$, i.e., 
\begin{equation}
\max_{\{\theta\}\in{\mathcal A}} \E \left[u(X_T^\theta) | X_0^\theta=x \right].
\label{maxproblem}
\end{equation}
In this paper we consider the case when the optimal response strategy $\theta$ is restricted by the unity  from above, i.e., ${\mathcal A} = \{ \{\theta\} \,|\, \theta_t \le 1, \; 0\le t \le T\}$. The function $u$ represents the terminal utility function. Throughout the paper we shall assume that $u$ is a strictly increasing and concave function, i.e., $u'(x)>0$ and $u''(x)<0$ for all $x\in\R$.

It follows from the theory of stochastic dynamic programming (see e.g. \cite{D}) that problem (\ref{maxproblem}) can be solved by introducing the so-called value function
\begin{equation}\label{VALUE}
V(x,t):=\sup_{\{\theta\}\in{\mathcal A}} \E[u(X_T^\theta) \,|\, X_t^\theta=x].
\end{equation}
Using the Bellman optimality principle combined with the tower law of conditioned 
expectations it can be shown that the value function satisfies the so-called 
Hamilton-Jacobi-Bellman (HJB) equation 
\begin{equation}
 \frac{\partial V}{\partial t}(x,t) + \max_{\{\theta\}\in{\mathcal A} } 
 \Big\{ \frac{(\sigma_t^\theta(x))^2}{2} \frac{\partial^2 V}{\partial x^2}(x,t) 
 +\mu_t^\theta(x)\frac{\partial V}{\partial x}(x,t) \Big\} = 0 , \quad 
V(x,T)=u(x), 
\label{HJB-gen}
\end{equation}
for all $x\in\R$ and $t\in[0,T)$ (see e.g. \cite{AI,D,IM,MS}). 

The main goal of this paper is to construct monotone traveling wave solutions to the HJB equation (\ref{HJB-gen}) subject to the constraint ${\mathcal A}= \{ \{\theta\} \,|\,  \theta_t \le 1, \; 0\le t\le T\}$ for the optimal decision policy $\{\theta\}$. Depending on the models considered, we show the existence of traveling wave solutions with positive as well as negative wave speeds. 

This paper is organized as follows. In the next section we investigate a simple 
HJB equation with drift and volatility functions linearly depending  on the optimal decision parameter $\theta$. Using a Riccati-like transformation we transform the HJB equation, originally stated for the value function $V$ into a fully nonlinear parabolic PDE for the reciprocal value of the optimal response function $\theta$. We construct a traveling wave solution with a decreasing wave profile, and we extend the results to the case when the underlying processes is governed by a SDE with a drift quadratically depending on the parameter $\theta$. In section 3 we investigate a more general HJB equation with a volatility function depending nonlinearly on the optimal decision policy parameter $\theta$. We again identify a range of model parameters for which a traveling wave solution exists and has a monotonically increasing profile.

\newpage
\section{Construction of a traveling wave solution to the HJB equation with a positive wave speed}
\label{sec:section2}

\noindent
In this section, we focus our attention to traveling wave solutions to the HJB equation. First, we shall examine a simplified model where the drift $\mu_t^\theta (X_t)$ and volatility $\sigma_t^\theta (X_t)$ are linear functions in $\theta$. Then we generalize the results to the HJB equation (\ref{HJB-gen}) with an underlying stochastic process satisfying a SDE (\ref{genproc}) with a drift function quadratically depending on the control parameter $\theta$.

\subsection{ A simple HJB equation}
\label{sec:simpleHJB}
\noindent

 In what follows, we shall analyze  a simplified model capturing  essential features  of a more complex model with general drift and volatility functions $\mu$ and  $\sigma$. We assume $X_t^\theta$ is a Brownian motion with drift $\mu^\theta = \omega \theta$ and volatility $\sigma^\theta=\theta>0$, i.e., 
\[
\d X_t^\theta = \omega \theta  \d t + \theta \d W_t,
\]
where $\omega>0$ is a positive parameter. We restrict our strategy $\theta$ by $1$ 
from above. Then the corresponding Hamilton--Jacobi--Bellman equation (\ref{VALUE}) for 
the value function $V(x,t)$ reads as follows:
\begin{equation}\label{HJB}
 \frac{\partial V}{\partial t}(x,t) + \sup_{\theta \leq 1} 
 \Big\{ \frac{1}{2}\theta^2 \frac{\partial^2 V}{\partial x^2}(x,t) 
 +\omega \theta\frac{\partial V}{\partial x}(x,t) \Big\} = 0 , \qquad 
V(x,T)=u(x) .  
\end{equation}
Suppose, for a moment, that (\ref{HJB}) has a classical solution $V$ such that 
$\partial_x V(x,t) > 0$, and $\partial^2_x V(x,t) < 0$ for all $x\in\R$ and $t\in[0,T]$. For justification of such an assumption we refer the reader to Proposition~\ref{prop:concavity}. Let us denote by $\theta^{*}(x,t)$ the optimal 
response strategy at $(x,t)$. It maximizes the function
\[
\R \ni \theta \mapsto \frac{1}{2}\theta^2 \frac{\partial^2 V}{\partial x^2}(x,t) 
 +\omega \theta\frac{\partial V}{\partial x}(x,t) \in \R
\]
subject to the constraint $\theta \le 1$. If the optimal response satisfies $\theta^{*}(x,t)<1$ then we have  
\begin{equation}\label{OP}
\theta^{*}(x,t)=- \omega \frac{\partial_x V(x,t)}{\partial^2_x V(x,t)} .  
\end{equation}
Placing (\ref{OP}) back into (\ref{HJB}) we obtain 
\begin{equation}\label{PDE}
\begin{split}
\frac{\partial V}{\partial t} 
 - \frac{\omega^2}{2}\frac{(\partial_x V)^2}{\partial^2_x V} 
 = 0  \quad\text{for  } 0 < t < T  , \qquad  
V(T,x) = u(x) .  
\end{split}
\end{equation}

Following \cite{AI} and \cite{MS} (see also \cite{IM}), we introduce the Riccati-like transformation
\begin{equation}\label{TR}
\varphi(x,t):= -\frac{1}{\omega} \frac{\partial^2_x V(x,t)}{\partial_x V(x,t)}.
\end{equation}
Recall that, in the context of optimal portfolio allocation problems, the function $\varphi$ is related to the so-called Arrow-Pratt coefficient of the 
absolute risk aversion (c.f. \cite{MMG}, \cite{P}). Performing straightforward 
calculations, the evolution equation for $\varphi$ becomes 
\begin{equation}\label{EQ1}
\frac{\partial\varphi}{\partial t} + \frac{1}{2}
\frac{\partial}{\partial x}\Big(\frac{1}{\varphi^2}
\frac{\partial\varphi}{\partial x}\Big) = 0.
\end{equation}

In view of the optimal response function $\theta^{*} = 1/\varphi$, equation (\ref{EQ1}) is fulfilled by $\varphi$ in the region $\{(x,t),\ \varphi(x,t) > 1\}$. On the other hand, if $\varphi(x,t) <1$, then the maximum in (\ref{HJB}) is attained at $\theta^* = 1$. Inserting $\theta=1$ into (\ref{HJB}) we end up with an equation for $V(x,t)$ of the form:
\begin{equation*}
 \frac{\partial V}{\partial t}
 + \frac{1}{2}\frac{\partial^2 V}{\partial x^2} 
 + \omega \frac{\partial V}{\partial x} = 0 , \qquad 
V(x,T)=u(x) .  
\end{equation*}
In terms of the transformed function $\varphi$, it can be further reduced to 
\begin{equation}\label{EQ2}
\frac{\partial\varphi}{\partial t} + \frac{1}{2}\frac{\partial}{\partial x}
\Big(\frac{\partial\varphi}{\partial x} - \omega (1-\varphi)^2 \Big) = 0 . 
\end{equation}
It is easy to see that equation (\ref{EQ2}) is satisfied in the region $\{(x,t),\ \varphi(x,t) < 1\}$. 

Combining equations (\ref{EQ1}) and (\ref{EQ2}) allows us to rewrite them in a compact form 
\begin{equation}\label{EQU}
\frac{\partial\varphi}{\partial t} + 
\frac{\partial^2}{\partial x^2} A(\varphi) 
+
\frac{\partial}{\partial x} B(\varphi)=0, \quad x\in\R,\ t\in(0,T),
\end{equation}
where  
\begin{equation}
A(\varphi)=\left\{
\begin{array}{cc}
\frac{1}{2}\varphi, &  \quad \hbox{for}\ \ \varphi \le 1, \\ 
1-\frac{1}{2\varphi}, & \quad \hbox{for}\ \ \varphi >1,
\end{array} 
\right.
\qquad
B(\varphi)=\left\{
\begin{array}{cc}
-\frac{\omega}{2} (1-\varphi)^2, &  \quad \hbox{for}\ \ \varphi \le 1, \\ 
0, & \quad \hbox{for}\ \ \varphi >1.
\end{array} 
\right.
\label{functionAB}
\end{equation}
Notice that $A(\varphi)$ and $B(\varphi)$ are increasing and $C^1$ continuous functions. 

\subsubsection{Construction of a traveling wave solution with a positive wave speed}
\label{sec:riskseeking}

In this subsection we shall construct a traveling wave solution to (\ref{EQU}) of 
the form 
\[
\varphi(x,t) = v(x+c(T-t)), \qquad x\in\R, \ t\in [0,T],
\]
with the  wave speed $c \in \R$ where $v=v(\xi)$ is a $C^1$ function defined on $\R$. 

Inserting the aforementioned ansatz on the solution $\varphi$ into  (\ref{EQU}) we conclude the existence of a constant $K_0\in \R$ such that the function $v=v(\xi)$ fulfills the identity
\begin{equation}
-c v(\xi) + \frac{d}{d\xi}(A(v(\xi))) +  B(v(\xi))  = K_0,
\label{travelwaveEQU}
\end{equation}
for all $\xi\in\R$. Let us introduce an auxiliary function $z=z(\xi)$ as follows:
\[
z(\xi) = A(v(\xi)).
\]
Then $\varphi(x,t) = v( x + c (T-t))$ is a traveling wave solution to (\ref{EQU}) if and only if the function $z$ is a solution to the ODE:
\begin{equation}
z^\prime(\xi) = F (z(\xi)),
\label{ode-riskseeking}
\end{equation}
where $F(z) = K_0 + c A^{-1}(z) - B( A^{-1}(z))$. Clearly,
\begin{equation}
F(z)=\left\{
\begin{array}{cc}
 K_0 +  2c z + \frac{\omega}{2} (1-2z)^2, &  \quad \hbox{for}\ \ z<\frac{1}{2}, \\ 
K_0 + \frac{c}{2(1-z)}, & \quad \hbox{for}\ \ \frac{1}{2}\le z<1.
\end{array} 
\right.
\label{functionF}
\end{equation}
Notice that the function $F$ is $C^1$ continuous for $z<1$. Its graph, for the case 
when $c>0$ and $K_0 + c <0$, is depicted in Fig.~\ref{fig:functionF}. In this case 
the function $F$ has exactly two roots $z^\pm$ such that $F(z^\pm)=0$ and 
$0<z^+< 1/2< z^-<1$, where
\begin{equation}
z^- =  1+ \frac{c}{2 K_0}, \qquad 
z^+ = \frac{1}{2}- \frac{c}{2\omega} - \frac{1}{2}\sqrt{c^2/\omega^2 - 2(c+K_0)/\omega}.
\label{zroots}
\end{equation}
We have $F^\prime(z^-) >0$ and $F^\prime(z^+) <0$ and $F(z)<0$ for 
$z^+<z<z^-$. Therefore, up to a shift in the argument $\xi$, there exists a unique solution $z=z(\xi)$ to (\ref{ode-riskseeking}) connecting the steady states $z^\pm$ such that 
\[
\lim_{\xi\to\pm \infty} z(\xi) = z^\pm, \quad 
z^+ < z(\xi) < z^-, \quad z^\prime(\xi) <0,\ \xi\in\R .
\]
The traveling wave profile $v=v(\xi)$ is given by $v(\xi) = A^{-1} (z(\xi))$. It satisfies:
\[
\lim_{\xi\to\pm \infty} v(\xi) = v^\pm, \quad 
v^+ < v(\xi) < v^-, \quad v^\prime(\xi) <0, \ \xi\in\R, 
\]
with $v^\pm = A^{-1}(z^\pm)$. Clearly $0<v^+< 1< v^-$. On the other hand, we can 
prescribe the limiting values $v^\pm$ and calculate the corresponding wave speed $c>0$ and constant $K_0$. Indeed, it is straightforward to verify that, for $0<v^+< 1< v^-$, the wave speed $c$ and $K_0$ are given by formulae:
\begin{equation}
c= \frac{\omega}{2} \frac{(1-v^+)^2}{v^- - v^+}, \qquad 
K_0 = -c v^-.
\label{speed}
\end{equation}

Since $F$ is a $C^1$ smooth nonlinear function we obtain that the solution $z$ is a $C^2$ smooth function in the $\xi$ variable. However, as $A$ is just $C^1$ smooth we obtain that the traveling wave profile $v$ is a $C^1$ smooth function in the $\xi$ variable only. As a consequence, we have that the solution $\varphi(x,t) = v(x + c (T-t))$ is $C^1$ smooth and it is a weak solution to (\ref{EQU}) in the usual sense. 

\begin{figure}
\begin{center}
\includegraphics[width=0.35\textwidth]{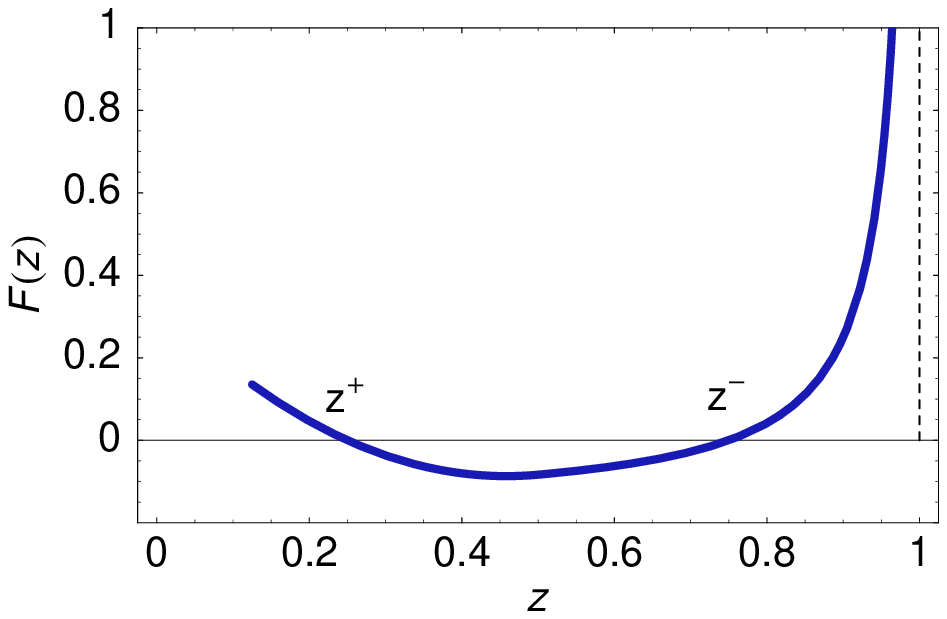}
\includegraphics[width=0.35\textwidth]{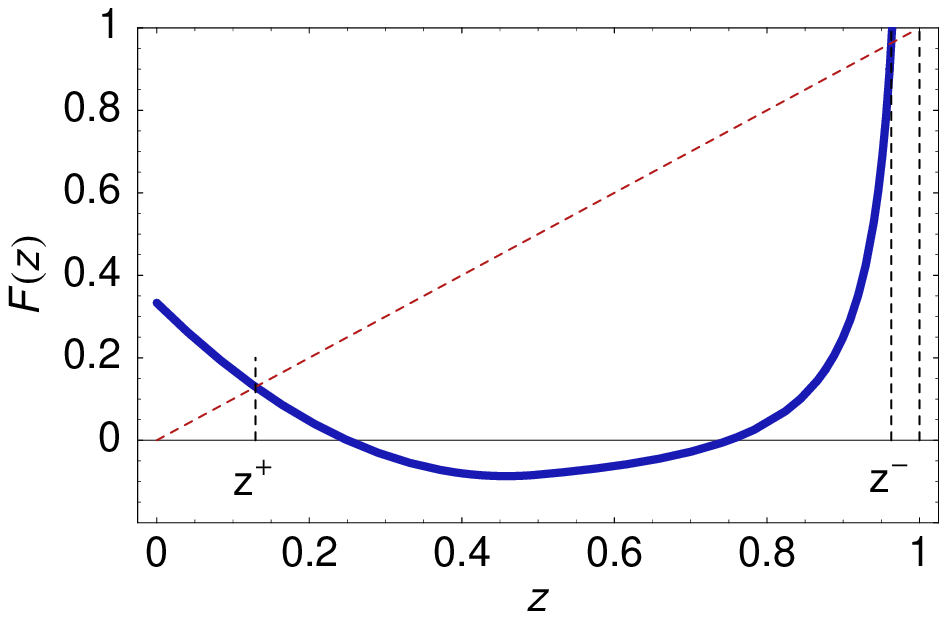}
\end{center}
\caption{
A graph of the function $F(z)$ in the case $c>0$ and $K_0+c<0$.
 The roots $z^\pm$ of the function $F(z)$ (left) and those of the function $\tilde F(z)=F(z)-z$ (right).
}
\label{fig:functionF} 
\end{figure}

\subsection{ A HJB equation for a quadratic drift function}
\label{sec:portfolioHJB}

\noindent
In this section  we shall assume the underlying stochastic process satisfying a SDE (\ref{genproc}) with a drift function $\mu^\theta$ quadratically depending on the parameter $\theta$, i.e. $\mu^\theta :=\omega\theta -\frac12 \theta^2$. The volatility is again assumed to be linear in $\theta$, $\sigma^\theta:=\theta$. The Hamilton--Jacobi--Bellman equation (\ref{VALUE}) for the value function $V$ has the form 
\begin{equation}\label{HJBX}
 \frac{\partial V}{\partial t}(x,t) + \sup_{\theta \leq 1} 
 \Big\{ \frac{1}{2}\theta^2 \frac{\partial^2 V}{\partial x^2}(x,t) 
 +\left(\omega \theta - \frac{1}{2}\theta^2\right) \frac{\partial V}{\partial x}(x,t) \Big\} = 0 , \quad 
V(x,T)=u(x), 
\end{equation}
In this case, the optimal response strategy $\tilde\theta^{*}(x,t)$ is given by 
\begin{equation}\label{OPX}
\tilde\theta^{*}(x,t)=- \omega \frac{\partial_x V(x,t)}{\partial^2_x V(x,t) - \partial_x V(x,t)}, 
\end{equation}
provided that $\tilde\theta^{*}(x,t)<1$. The function $V$ solves the nonlinear PDE:
\begin{equation}\label{PDEX}
\begin{split}
\frac{\partial V}{\partial t} 
 - \frac{\omega^2}{2}\frac{(\partial_x V)^2}{\partial^2_x V - \partial_x V} 
 = 0  \quad\text{for  } 0 < t < T  , \quad  
V(T,x) = u(x),\ x\in\R .  
\end{split}
\end{equation}
Applying the Riccati-like transformation
\begin{equation}\label{TRX}
\tilde \varphi(x,t):= -\frac{1}{\omega} \frac{\partial^2_x V(x,t) - \partial_x V(x,t)}{\partial_x V(x,t)} 
 = - \frac{1}{\omega} \left( \frac{\partial^2_x V(x,t)}{\partial_x V(x,t)} - 1 \right), 
\end{equation}
it is easy to verify that the equation for $\tilde \varphi$ reads as follows: 
\begin{equation}\label{EQ1X}
\frac{\partial\tilde\varphi}{\partial t} + \frac{1}{2}
\frac{\partial}{\partial x}\Big(\frac{1}{\tilde\varphi^2}
\frac{\partial\tilde\varphi}{\partial x}\Big) -\frac{\partial}{\partial x} \frac{1}{2\tilde\varphi}= 0,
\end{equation}
and it is satisfied by $\tilde\varphi$ in the region $\{(x,t),\ \tilde\varphi(x,t) > 1\}$. If  $\tilde\varphi(x,t) <1$, then the maximum in (\ref{HJBX}) is attained at $\tilde\theta^* = 1$. In this case, the equation for the solution $V(x,t)$ and its Riccati transformation $\tilde\varphi$ are as follows:
\begin{equation*}
 \frac{\partial V}{\partial t}
 + \frac{1}{2}\frac{\partial^2 V}{\partial x^2} 
 + (\omega-\frac{1}{2}) \frac{\partial V}{\partial x} = 0 , \qquad 
V(x,T)=u(x),
\end{equation*}
and
\begin{equation}\label{EQ2X}
\frac{\partial\tilde\varphi}{\partial t} + \frac{1}{2}\frac{\partial}{\partial x}
\Big(\frac{\partial\tilde\varphi}{\partial x} - \omega (1-\tilde\varphi)^2 + \tilde\varphi\Big) = 0, 
\end{equation}
provided that $\tilde\varphi(x,t) < 1$. Hence we can rewrite the equation for $\tilde\varphi$ as follows:
\begin{equation}\label{EQUX}
\frac{\partial\tilde\varphi}{\partial t} + 
\frac{\partial^2}{\partial x^2} \tilde A(\tilde\varphi) 
+
\frac{\partial}{\partial x} \tilde B(\tilde\varphi)=0, \quad x\in\R,\ t\in(0,T),
\end{equation}
where $\tilde A(\varphi)=A(\varphi)$ and $\tilde B(\varphi)= B(\varphi) + A(\varphi)$ ($A$ and $B$ are defined in (\ref{functionAB})).

Next, following the analysis from section~\ref{sec:riskseeking}, we can construct a traveling wave solution to (\ref{EQUX}) of the form 
\[
\tilde\varphi(x,t) = \tilde v(x+\tilde c(T-t)), \qquad x\in\R, \ t\in [0,T],
\]
 with the wave speed $\tilde c \in \R$ and the profile $\tilde v=\tilde v(\xi)$. Since $\tilde A \equiv A$ and $\tilde B \equiv B + A$ the transformed wave profile $\tilde z(\xi) = A(\tilde v(\xi))$ should satisfy the ODE:
\begin{equation}
\tilde z^\prime(\xi) = \tilde F (\tilde z(\xi)),
\label{ode-riskseekingX}
\end{equation}
where $\tilde F(z) = F(z) - z$ and $F$ is defined by (\ref{functionF}), i.e. $F(z)=K_0 +\tilde c A^{-1}(z) - B( A^{-1}(z) )$. Here $K_0$ is a constant. 
In Fig.~\ref{fig:functionF} (right) we plot the graph of a function $F(z)$ for  $\tilde c>0$ and $K_0 + \tilde c <0$. In such a situation, the function $\tilde F$ has exactly two roots $0<\tilde z^+< 1/2< \tilde z^-<1$ such that $F(\tilde z^\pm)=\tilde z^\pm$. Furthermore, $\tilde F^\prime(\tilde z^-) >0$ and $\tilde F^\prime(\tilde z^+) <0$. As in the previous section~\ref{sec:riskseeking}, there exists, up to a shift in the argument $\xi$, a unique solution $\tilde z=\tilde z(\xi)$ connecting the steady states, i.e. $\lim_{\xi\to\pm \infty} \tilde z(\xi) = \tilde z^\pm$. The corresponding traveling wave profile $\tilde v=\tilde v(\xi)$ given by $\tilde v(\xi) = A^{-1} (\tilde z(\xi))$ satisfies:
\[
\lim_{\xi\to\pm \infty} \tilde v(\xi) = \tilde v^\pm, \quad 
\tilde v^+ < \tilde v(\xi) < \tilde v^-, \quad \tilde v^\prime(\xi) <0, \quad  \xi\in \R, 
\]
where $\tilde v^\pm = A^{-1}(\tilde z^\pm)$, $0<\tilde v^+< 1< \tilde v^-$. Again we can prescribe the limiting values $0<\tilde v^+< 1< \tilde v^-$ and calculate the corresponding wave speed $\tilde c>0$ and constant $K_0$. Since $0=\tilde F(\tilde z^\pm) = K_0 +  \tilde c \tilde v^\pm  - \tilde B(\tilde v^\pm)$ we have $\tilde c = (\tilde B(\tilde v^-) - \tilde B(\tilde v^+))/(\tilde v^- - \tilde v^+)$ and $K_0= -\tilde c \tilde v^- + \tilde B(\tilde v^-)$. Since $\tilde B \equiv B + A$ we have
\begin{equation}
\tilde c=  \frac{\frac{\omega}{2} (1-\tilde v^+)^2 + 1 -\frac{1}{2\tilde v^-} -\frac{\tilde v^+}{2} }{\tilde v^- - \tilde v^+}. 
\label{speedX}
\end{equation}

Summarizing the results of this section we conclude the following theorem.

\begin{theorem}\label{theorem:riskseeking}
For any limiting values $0<v^+< 1< v^-$ there exists a speed $c>0$ given by (\ref{speed}) 
such that the Hamilton-Jacobi-Bellman equation (\ref{HJB}) with a range bound 
$\{\theta\le 1\}$ has a solution $V(x,t)$ such that the optimal response function 
$\theta^*(x,t)$ given by (\ref{TR}) has the form of $\theta^* = \min(1, 1/\varphi)$ where  $\varphi(x,t) = v(x+c(T-t))$ and $v(\xi)$ is a $C^1$ smooth and strictly decreasing function, $\lim_{\xi\to\pm \infty} v(\xi) = v^\pm$.

The statement remains true if we consider the HJB equation (\ref{HJBX}) for the reduced optimal portfolio selection problem with the traveling wave profile $\tilde \varphi(x,t) = \tilde v(x+\tilde c(T-t))$, $\lim_{\xi\to\pm \infty} \tilde v(\xi) = \tilde v^\pm$ and the optimal response function $\tilde \theta^* = \min(1, 1/\tilde \varphi)$. Here the wave speed $\tilde c>0$ is given by (\ref{speedX}) for any prescribed limiting values $0<\tilde v^+< 1< \tilde v^-$.

\end{theorem}

\begin{proposition}
For HJB equations (\ref{HJB}) and (\ref{HJBX}) there is no traveling wave profile with negative wave speed.
\end{proposition}

\begin{proof}
Indeed, for the traveling wave speed we have  $c = (B(v^-) - B(v^+))/(v^- - v^+) >0$ because the function $B$ is increasing. Analogously, $\tilde c>0$ as $\tilde B = B + A$ is an increasing function. Moreover, the profile $v(\xi)$ is always a decreasing function, as there are no roots $0<z^- <\frac12 < z^+ <1$ of $F$ such that $F^\prime(z^-)>0$ and $F^\prime(z^+)<0$. The latter follows from the inequality $F(0) =K_0 +\frac{\omega}{2} > K_0 + c = F(\frac12)$ and $F(1^-) = -\infty$ in the case $c<0$. The same statement holds true for the profile $\tilde v$. 
\end{proof}

\begin{figure}
\begin{center}
\includegraphics[width=0.35\textwidth]{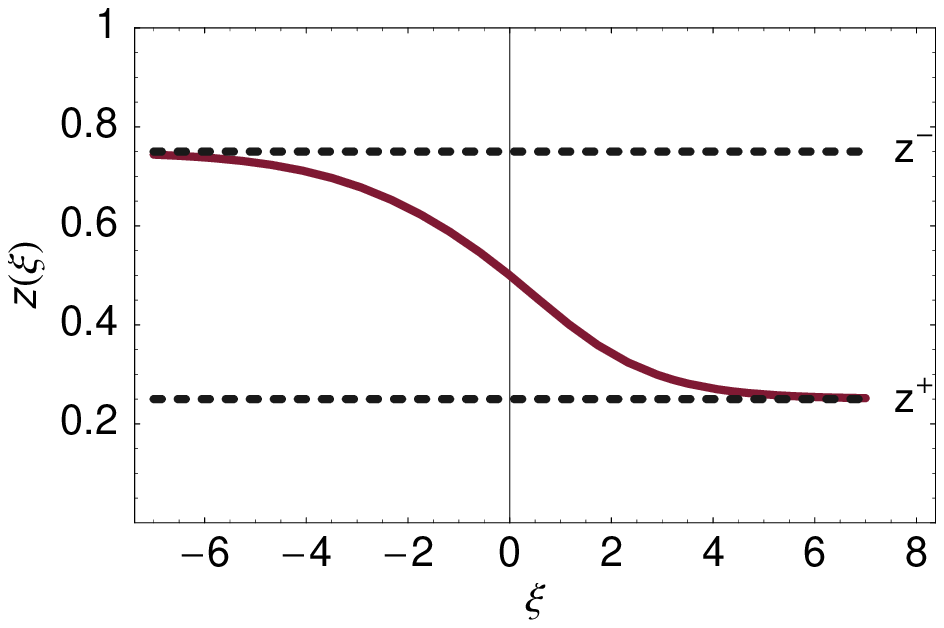}
\includegraphics[width=0.35\textwidth]{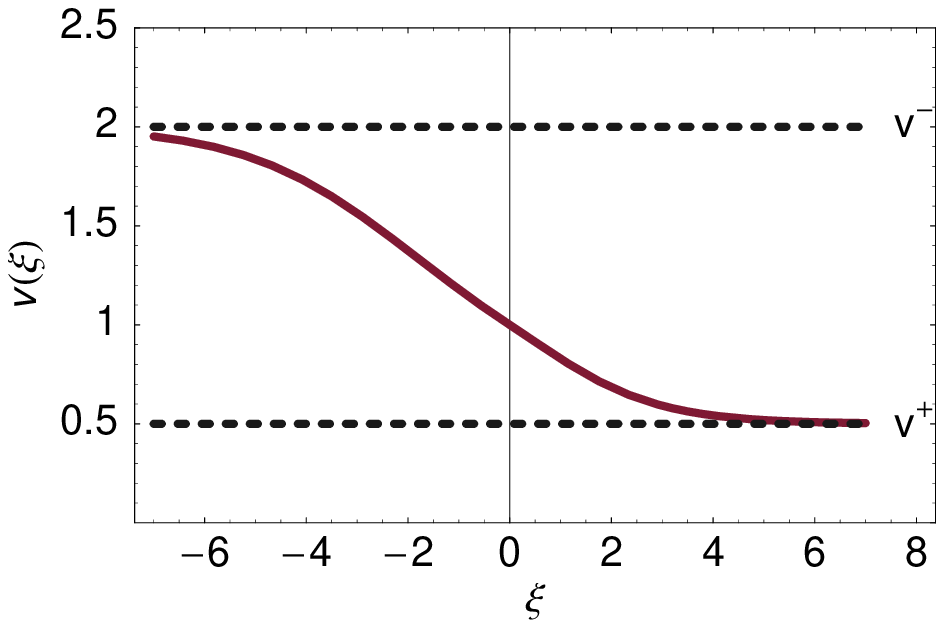}
\end{center}
\caption{
A graph of the function $z$ (left) and the traveling wave profile $v$ (right) for the parameter values $\omega=1, v^+=0.5, v^-=2$.}
\label{fig:travelwave} 
\end{figure}

\begin{figure}
\begin{center}
\includegraphics[width=0.45\textwidth]{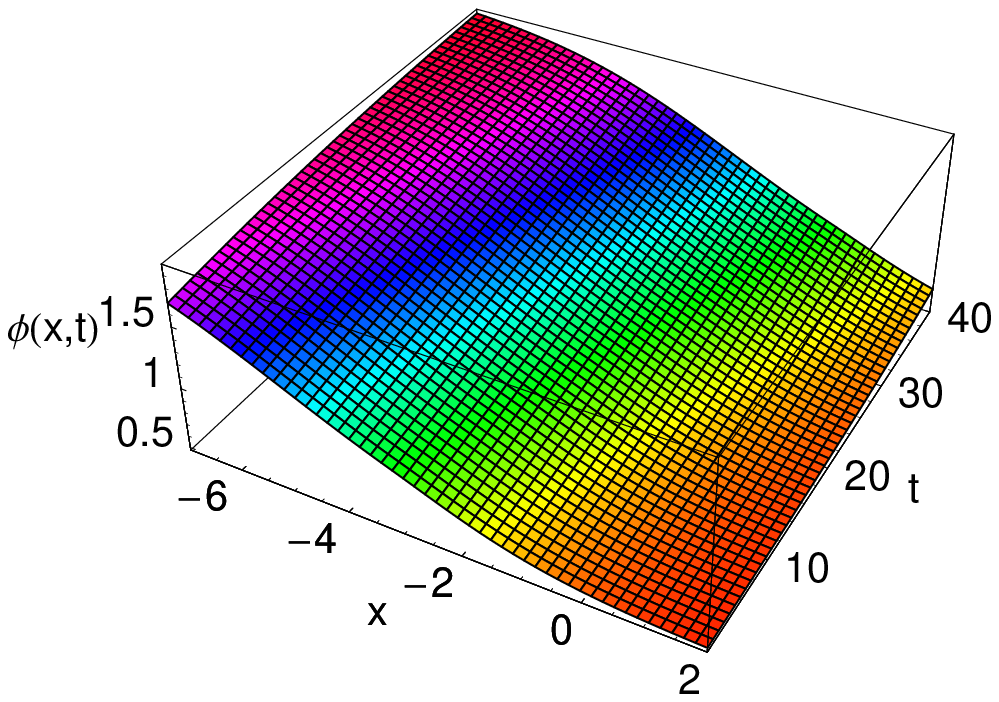}
\includegraphics[width=0.45\textwidth]{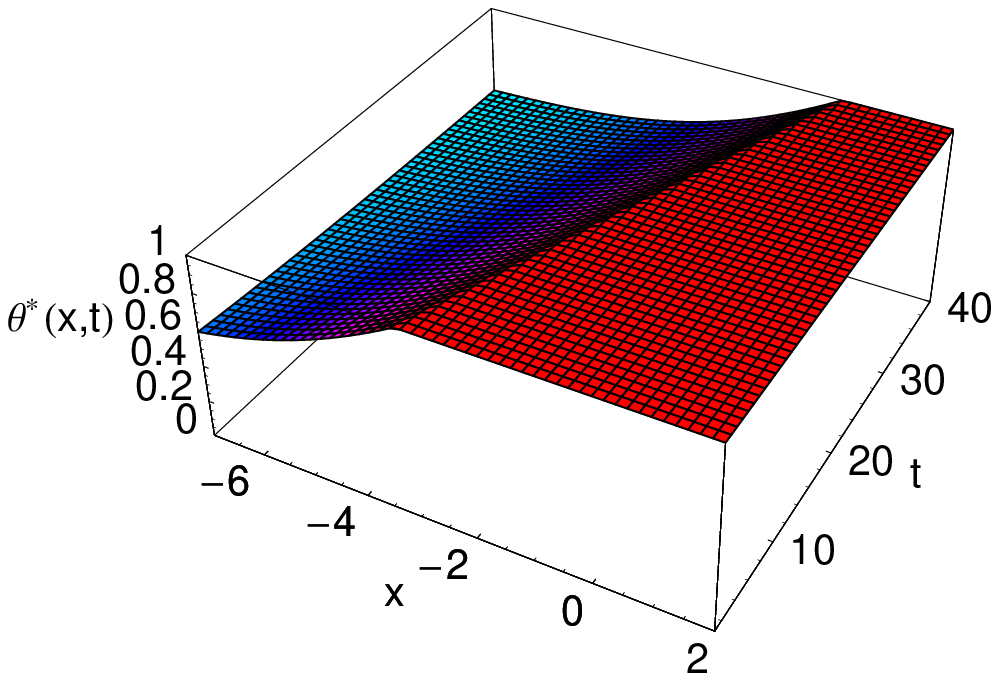}

\end{center}
\caption{
Graphs of the function $\phi(x,t)$ (left) and the response function $\theta^*(x,t)$ (right).}
\label{fig:phitheta} 
\end{figure}

\section{A HJB equation with more general drift and volatility functions}

One disadvantage of the HJB equations with a bounded constraint discussed in 
 sections \ref{sec:simpleHJB} and \ref{sec:portfolioHJB}  
is the fact that there is no traveling wave solution with 
increasing wave profile $v$. As a consequence, the optimal response function 
$\theta^*(x,t)$ is nondecreasing with respect to $x$. In this section we analyze a HJB equation with more general drift and volatility functions. Under suitable assumptions made on the model parameters,  we shall prove that there exists a traveling wave solution having an increasing wave profile $v$ and negative wave speed $c<0$.

We shall assume that a stochastic process  $X_t^\theta$ follows a Brownian motion 
\[
dX_t^\theta = \mu^\theta dt + \sigma^\theta dW_t
\]
with drift $\mu^\theta$  and volatility  $\sigma^\theta$ given by
\[
\mu^\theta = \beta+\omega \theta, \qquad (\sigma^\theta)^2 = 2(\alpha^2+\frac{1}{m}|\theta|^m),
\] 
where $\alpha,\beta,\omega, m\in {\mathbb R}, m>1$, and $\omega>0$ are model parameters. We again restrict our response strategy $\theta$ by $1$ from above. The corresponding HJB equation (\ref{VALUE}) for the value function $V(x,t)$ reads as follows:

\begin{eqnarray}\label{generalHJB}
 && \frac{\partial V}{\partial t}(x,t) + \sup_{\theta \leq 1} 
 \Big\{ \big(\alpha^2+\frac{1}{m}|\theta|^m\big) \frac{\partial^2 V}{\partial x^2}(x,t) 
 +(\beta+\omega \theta)\frac{\partial V}{\partial x}(x,t) \Big\} = 0, \nonumber
\\ 
&& V(x,T)=u(x). 
\end{eqnarray}
The simplified model discussed in the previous section corresponds to the choice 
$\alpha=\beta=0$, and $m=2$. 

Again, supposing $\partial_x V>0 $ and $\partial^2_x V <0$ the unconstrained  optimal response 
strategy $\bar\theta$ at $(x,t)$ is the unique argument of the maximum of the function
\[
\theta \mapsto 
\big(\alpha^2+\frac{1}{m}|\theta|^m\big) \partial^2_x V  +(\beta+\omega \theta)\partial_x V.
\]
Therefore
\[
\bar\theta = \varphi^{-\frac{1}{m-1}},
\]
where we have again employed  the new variable $\varphi$ defined by means of the Riccati-like transformation:
\[
\varphi(x,t):= -\frac{1}{\omega} \frac{\partial^2_x V(x,t)}{\partial_x V(x,t)}.
\]
Now, if $\varphi(x,t) >1$ at $(x,t)$ then $\bar\theta(x,t)<1$ and therefore, for the optimal response $\theta^*(x,t)$, we have $\theta^*(x,t) = \bar\theta(x,t)$. Hence the value function $V(x,t)$ satisfies 
\[
\frac{\partial V}{\partial t} + (\alpha^2 + \frac{1}{m} \varphi^{-\frac{m}{m-1}}) 
\frac{\partial^2 V}{\partial x^2} + (\beta+\omega \varphi^{-\frac{1}{m-1}}) 
\frac{\partial V}{\partial x} =0.
\]
Notice the following recurrent relations:
\begin{equation}
\partial^2_x V = -\omega \varphi  \partial_x V, \qquad 
\partial^3_x V = (\omega^2 \varphi^2 - \omega \partial_x\varphi)  \partial_x V.
\label{recurrentV}
\end{equation}
Using the relation for $\partial^2_x V$ we can rewrite the equation for $V$ in the form:
\begin{equation}
\frac{\partial V}{\partial t}  + g \frac{\partial V}{\partial x} =0,
\quad\hbox{where}\ \ 
g = \beta -\omega \alpha^2 \varphi +\omega \frac{m-1}{m}  \varphi^{-\frac{1}{m-1}}.
\label{Vgequ}
\end{equation}
Since 
\begin{equation}
\frac{\partial \varphi}{\partial t}  = -\frac{1}{\omega} \frac{\partial^3_{xxt} V}{\partial_x V}
+ \frac{1}{\omega} \frac{\partial^2_{xx} V \partial^2_{xt} V }{(\partial_x V)^2}
=
\frac{\partial \varphi}{\partial t}  = -\frac{1}{\omega} \frac{\partial^3_{xxt} V}{\partial_x V}
- \varphi \frac{\partial^2_{xt} V }{\partial_x V}
\label{phiequ}
\end{equation}
we obtain from (\ref{Vgequ}) and (\ref{recurrentV}) the equation 
\begin{equation}
\frac{\partial \varphi}{\partial t}  = \frac{1}{\omega} \frac{\partial^2 g}{\partial x^2}
-
\frac{\partial}{\partial x}(g \varphi),
\label{gequ}
\end{equation}
which can be further rewritten in terms of the function $\varphi$ as follows:
\begin{equation}\label{generalEQ1}
\frac{\partial\varphi}{\partial t} + \frac{\partial^2}{\partial x^2}\left( 
\alpha^2 \varphi  - \frac{m-1}{m} \varphi^{-\frac{1}{m-1}}
\right)
+ \frac{\partial}{\partial x}
\Big(\beta \varphi - \omega \alpha^2 \varphi^2  + 
\frac{m-1}{m}\omega \varphi^{\frac{m-2}{m-1}}\Big) = 0,
\end{equation}
provided that $\varphi(x,t) >1$. 

On the other hand, if $\varphi(x,t) <1$ at $(x,t)$ then $\bar\theta(x,t)>1$ and 
therefore for the optimal response $\theta^*(x,t)$ which is restricted by $1$ from above we have $\theta^*(x,t) = 1$. Then the value function $V(x,t)$ satisfies the following equation:
\[
\frac{\partial V}{\partial t} + (\alpha^2 + \frac{1}{m}) 
\frac{\partial^2 V}{\partial x^2} + (\beta+\omega) \frac{\partial V}{\partial x} =0.
\]
In view of relations (\ref{recurrentV}) we have
\[
\frac{\partial V}{\partial t}  + \tilde g \frac{\partial V}{\partial x} =0,
\quad\hbox{where}\ \ 
\tilde g = \beta + \omega -\omega( \alpha^2 + \frac{1}{m})  \varphi.
\]
Using (\ref{phiequ}), (\ref{recurrentV}) and (\ref{gequ}) (with $g$ replaced by 
$\tilde g$) we finally obtain the following reaction diffusion equation
\begin{equation}\label{generalEQ2}
\frac{\partial\varphi}{\partial t} 
+ (\alpha^2 + \frac{1}{m})\frac{\partial^2\varphi}{\partial x^2}
+ \frac{\partial}{\partial x}\Big((\beta+\omega)\varphi 
- \omega (\alpha^2 + \frac{1}{m}) \varphi^2  \Big) = 0,
\end{equation}
which is fulfilled by $\varphi$ in the case when $\varphi(x,t) <1$.

Similar to the simplified problem ($\alpha=\beta=0$, and $m=2$) in both cases 
$\varphi>1$ and $\varphi<1$ we obtain that the function $\varphi$ is a solution 
to (\ref{EQU}), that is:
\begin{equation}
\frac{\partial\varphi}{\partial t} + 
\frac{\partial^2}{\partial x^2} A(\varphi) 
+
\frac{\partial}{\partial x} B(\varphi)  =0,
\label{generalEQU}
\end{equation}
where  
\begin{equation}
A(\varphi)=\left\{
\begin{array}{cc}
(\alpha^2 +\frac{1}{m})\varphi, &  \quad \hbox{for}\ \ \varphi \le 1, 
\\ 
\\
1-\frac{m-1}{m}\varphi^{-\frac{1}{m-1}} + \alpha^2 \varphi, & \quad \hbox{for}\ \ \varphi >1,
\end{array} 
\right.
\label{functionAgen}
\end{equation}
\begin{equation}
B(\varphi)=\left\{
\begin{array}{cc}
(\beta+\omega)\varphi - \omega (\alpha^2 + \frac{1}{m}) \varphi^2 
- \omega\frac{m-1}{m}, &  \quad \hbox{for}\ \ \varphi \le 1, 
\\ 
\\
\beta \varphi - \omega \alpha^2 \varphi^2  
+ \frac{m-1}{m}\omega (\varphi^{\frac{m-2}{m-1}} - 1), & \quad \hbox{for}\ \ \varphi >1.
\end{array} 
\right.
\label{functionBgen}
\end{equation}
Notice that we have added appropriate constants to the definition of $B$ in order to make it continuous across the value $\varphi=1$. Both functions $A$ and $B$ are 
$C^1$ continuous functions for $\varphi>0$. Moreover, the function $A$ is strictly increasing.

\subsection{A traveling wave solution with negative wave speed}
\label{sec:riskaverse}

The aim of this subsection is to show, for suitable model parameters 
$\alpha,\beta,\omega$ and $m$, that there exists a traveling wave solution $\varphi$ to (\ref{EQU}) with increasing wave profile $v$. To this end, we search for a traveling wave solution to (\ref{EQU}) of the form 
\[
\varphi(x,t) = v(x+c(T-t)), \qquad x\in\R,\  t\in [0,T],
\]
with negative wave speed $c < 0$ and where $v=v(\xi)$ is a $C^1$ function defined 
on $\R$. Inserting the traveling wave form of a solution $\varphi$ into (\ref{EQU}) 
we obtain that the function $v=v(\xi)$ should fulfill identity (\ref{travelwaveEQU}). Again, we introduce an auxiliary 
function $z(\xi) = A(v(\xi))$. With this substitution, $\varphi(x,t) = v( x + c (T-t))$ 
is a traveling wave solution to (\ref{EQU}) if and only if the function $z$ is a solution 
to the ODE:
\begin{equation}
z^\prime(\xi) = F (z(\xi)), \qquad \hbox{where}\ \ F(z) =  G(A^{-1}(z)), 
\quad G(v) = K_0 + c v  - B(v).
\label{generalZODE}
\end{equation}
Notice that the function $G$ is $C^1$ continuous for $v>0$. The function $A(v)$ is 
increasing for $v>0$ and its range is the interval $(0,1)$ if $\alpha=0$, and 
$(0,\infty)$ if $\alpha\not=0$. In order to construct an increasing 
traveling wave profile $v$ such that $\lim_{t\to\pm\infty}v(\xi) = v^\pm$ where 
$0<v^-< 1 < v^+$, we first have to find roots $z^\pm= A(v^\pm)$ of the function 
$F$ such that $F(z)>0$ for $z^-<z<z^+$ and $F^\prime(z^-)>0, F^\prime(z^+)<0$. 
To this end, it is sufficient to investigate the roots and behavior of the function 
$G(v)$ for positive $v>0$. Clearly, if $m\ge 2$ then the function $B$ is 
concave, $B''(v) \le 0, v\not=1$. Thus $G$ is convex and there are no roots $v^\pm$ 
such that $G(v) >0$ for $v^-<v<v^+$. For this reason we cannot find the desired form of a 
traveling wave solution $\varphi$ for the case $m\ge 2$. 

\begin{figure}
\begin{center}
\includegraphics[width=0.35\textwidth]{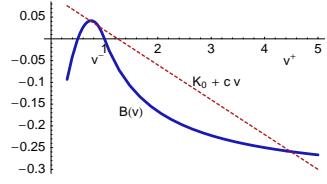}
\end{center}
\caption{
A graph of the functions $B(v)$ and $K_0 + c v$ for model parameters: 
$m=3/2, \alpha=\beta=0, c=-K_0=-0.1, \omega=1$.}
\label{fig:Broots} 
\end{figure}

The function $v\mapsto B(v)$ is strictly concave for $v\le1$. It has a local maximum at $\underline{v}^- = \frac{\beta+\omega}{2\omega(\alpha^2+1/m)}$ and it is strictly convex for $v>1$ provided that $\omega >0, \alpha=0$ and $1<m<2$. Notice that, for $\alpha=\beta=0, 1<m<2$ the maximum of $B$ is attained at $\underline{v}^- = m/2$. In what follows, we restrict ourselves to the case $\alpha=\beta=0$, and $m\in(1,2)$. Inspecting the behavior of the function $B$ (see also Fig.~\ref{fig:Broots}) we conclude the following result. 

\begin{theorem}\label{risk-averse}
Assume that $\alpha=\beta=0, m\in(1,2)$, and $\omega>0$. Let the limits $v^\pm$ be such that $\underline{v}^-  < v^-< 1$ and $v^+ > \underline{v}^+ >1$ where $\underline{v}^- = m/2<1$ and $\underline{v}^+>1$ is the unique root of the secant equation $B(\underline{v}^+) - B(v^-) = B^\prime(v^-) (\underline{v}^+ - v^-)$.

\begin{enumerate}
\item there exists a traveling wave speed $c<0$ given by 
\[
c= \frac{B(v^+) - B(v^-)}{v^+ - v^-}
\]
and intercept $K_0=B(v^-) - c v^-$ such that 
$G(v^\pm)=0$ and $G^\prime(v^-)>0$ and $G^\prime(v^+)<0$;

\item there exists a solution $V(x,t)$ to HJB (\ref{generalHJB}) with the optimal 
response function $\theta^*(x,t)$ given by $\theta^* = \min(1, \varphi^{-\frac{1}{m-1}} )$  where  $\varphi(x,t) = v(x+c(T-t))$ and $v(\xi)$ is a $C^1$ smooth and strictly increasing function, $\lim_{\xi\to\pm \infty} v(\xi) = v^\pm$.

\end{enumerate}

\end{theorem}

\begin{proof}
Part 1) follows from the behavior of the function $B$ given by (\ref{functionBgen})  for parameter values $\alpha=\beta=0, m\in(1,2)$ and $\omega>0$. Notice that the condition $v^+ > \underline{v}^+$, where $\underline{v}^+$ is the unique root of the secant equation $B(\underline{v}^+) - B(v^-) = B^\prime(v^-) (\underline{v}^+ - v^-)$, is necessary in order to find two roots of the equation $K_0 + c v = B(v)$ where $c > B^\prime(v^-)$ with the property $G(v^-) \equiv c - B^\prime(v^-) >0$ and $G(v^+) \equiv c - B^\prime(v^+) <0$. 

Part 2). Since the function $A$ is strictly increasing and $G^\prime(v^-)>0$ and 
$G^\prime(v^+)<0$ the ODE (\ref{generalZODE}) has a solution $z=z(\xi)$ such  that $F(z^\pm)=0$ and $F(z)>0$ for $z^-<z<z^+$ and $F^\prime(z^-)>0, F^\prime(z^+)<0$  where  $z^\pm= A(v^\pm)$. 

Then there exists a traveling wave solution $\varphi$ to (\ref{EQU}) of the form 
$\varphi(x,t) = v(x+c(T-t))$ with negative wave speed $c < 0$ where 
$v(\xi)= A^{-1}(z(\xi))$. Moreover, $\lim_{\xi\to\pm \infty} v(\xi) = v^\pm = A^{-1}(z^\pm)$. 
If $\varphi(x,t)>1$ then the optimal response $\theta^*(x,t)$ is given by 
$\theta^*(x,t) = \bar\theta(x,t) = \varphi(x,t)^{-\frac{1}{m-1}}$. On the other 
hand, $\theta^*(x,t)=1$ provided that $\varphi(x,t)\le1$, which concludes the proof. 
\end{proof}

By a non-monotone traveling wave solution we mean a non-constant solution whose derivative changes the sign several times.

\begin{proposition}
There are no non-monotone traveling wave solutions to the fully nonlinear equations (\ref{EQU}) and (\ref{generalEQU}) analyzed in sections \ref{sec:riskseeking} and \ref{sec:riskaverse}, respectively. 
\end{proposition}
\begin{proof}
We follow the ideas from \cite{IM}. If we assume to the contrary $z^\prime(\xi_0) =0$ then, as a consequence of the uniqueness of solutions to ODEs (\ref{generalZODE}) and (\ref{ode-riskseeking}), we obtain $z(\xi) \equiv const$ for all $\xi\in\R$. Consequently, the profile $v$ is constant, a contradiction. 
\end{proof}

We finish this section with two computational examples. In Fig.~\ref{fig:functionG} 
we plot a graph of the function $G(v)$ for the model parameters: 
$m=3/2, \alpha=\beta=0, c=-K_0=-0.1$ and $\omega=1$. In this case we 
have $v^-=1$ is a root of $G$. Moreover, $B^\prime(1) = \omega \frac{m-2}{m}$ 
and so $G^\prime(1) = c - \omega \frac{m-2}{m}$. Hence $G^\prime(v^-) >0$ if and only 
if $c > - \omega \frac {2- m}{m}$. The function $G$ also has a root $v^+=10/3 >1$. Since $m=3/2$ and $K_0+ c =0$ 
we have the analytic expression for $v^+=\frac{\omega}{-3c} >1$. Moreover, $G^\prime(v^+) = c  + 3 c^2<0$ for $\omega=1, c=-K_0=-0.1$, and $m=3/2$.
In Fig.~\ref{fig:travelwaveAdverse}  we plot a graph of the function $z$  and the traveling wave profile $v$. The functions $\phi(x,t)$ and the optimal response function  $\theta^*(x,t)$ are depicted in Fig.~\ref{fig:phithetaAdverse}.

\begin{figure}
\begin{center}
\includegraphics[width=0.35\textwidth]{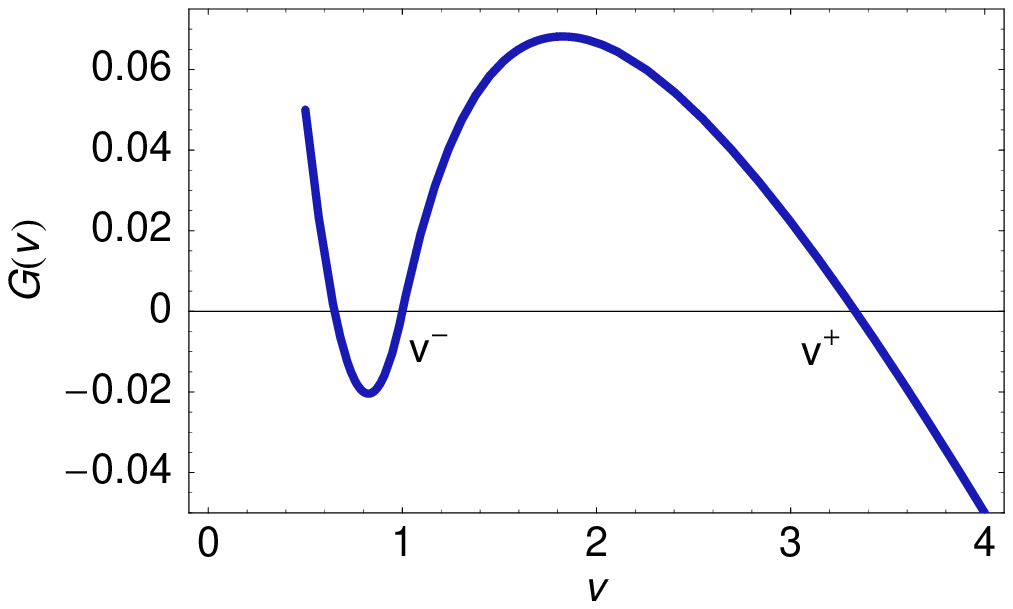}
\includegraphics[width=0.35\textwidth]{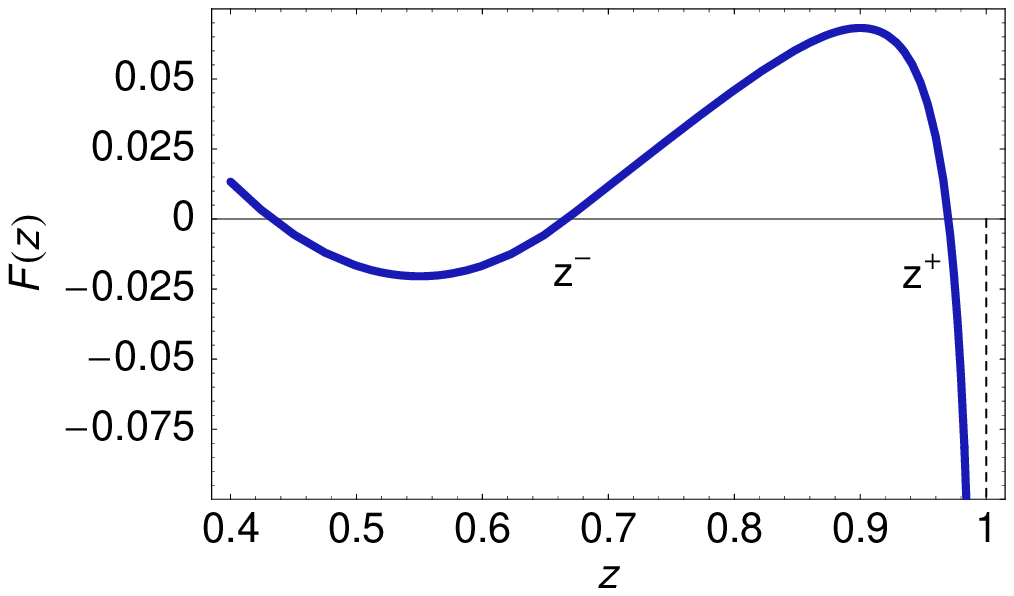}
\end{center}
\caption{
A graph of the function $G(v)$ for model parameters: 
$m=3/2, \alpha=\beta=0, c=-K_0=-0.1, \omega=1$ (left). In this case 
$v^-=1, v^+=10/3$. A graph of the function $F(z)=G(A^{-1}(z))$ (right) with 
roots $z^-=2/3, z^+ \approx 0.97$.}
\label{fig:functionG} 
\end{figure}

\begin{figure}
\begin{center}
\includegraphics[width=0.35\textwidth]{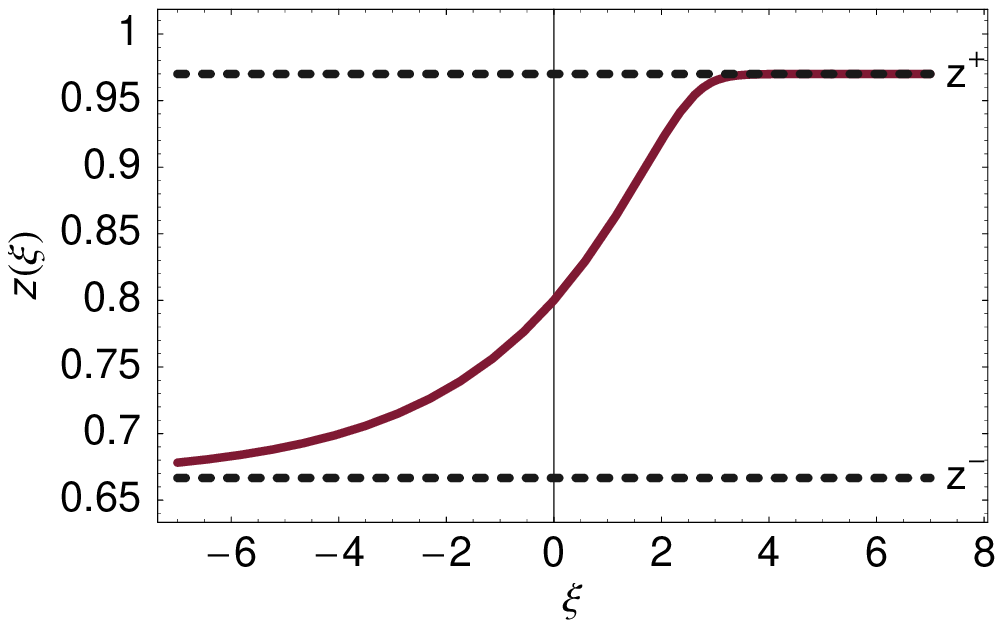}
\includegraphics[width=0.35\textwidth]{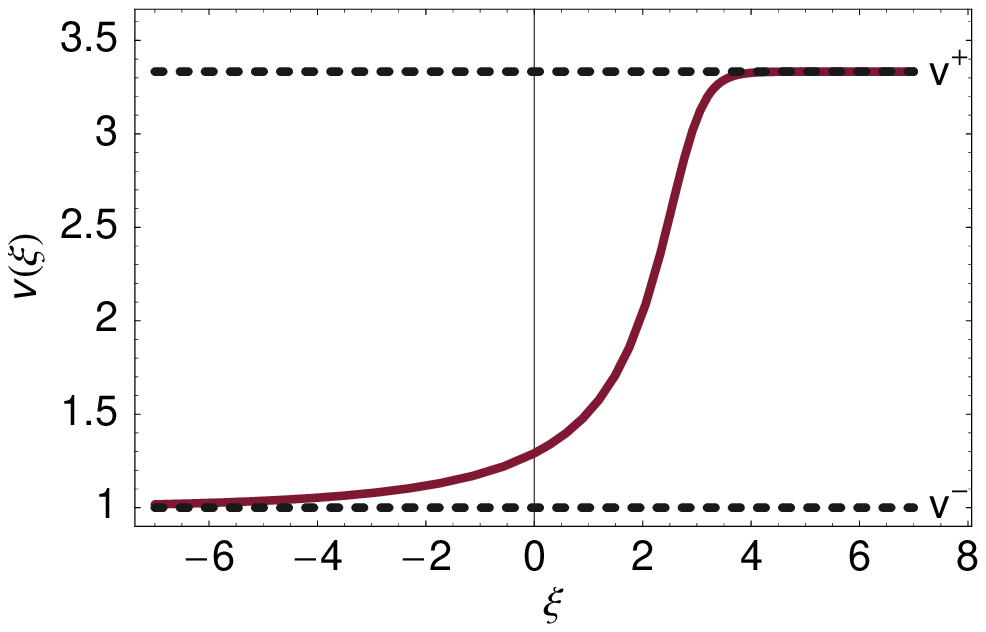}
\end{center}
\caption{
A graph of the function $z$ (left) and the traveling wave profile $v$ (right) 
for the parameter values 
$m=3/2, \alpha=\beta=0, c=-K_0=-0.1, \omega=1$.}
\label{fig:travelwaveAdverse} 
\end{figure}

\begin{figure}
\begin{center}
\includegraphics[width=0.45\textwidth]{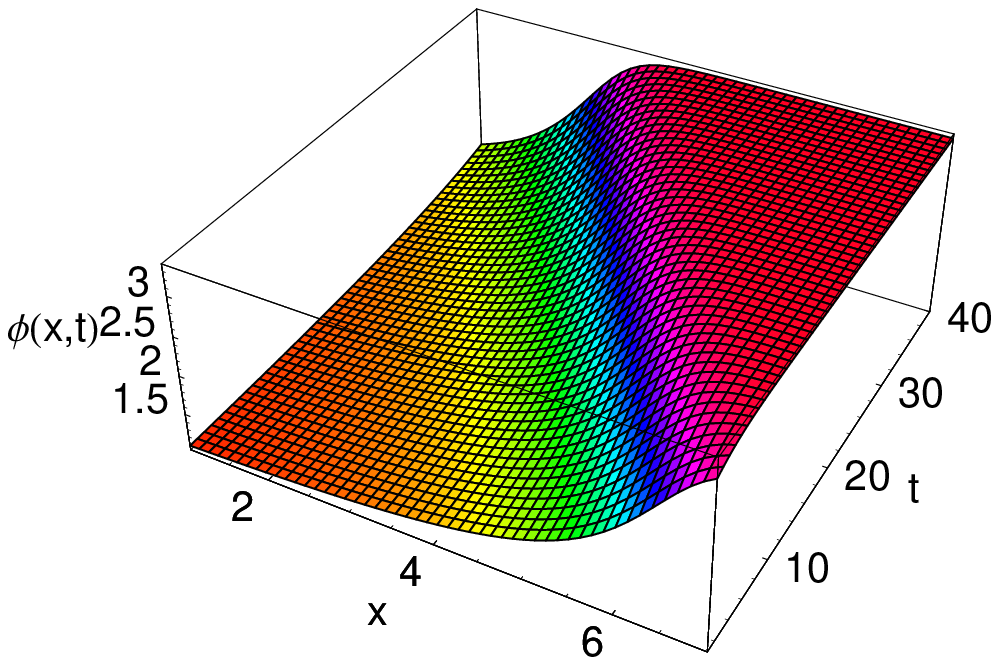}
\includegraphics[width=0.45\textwidth]{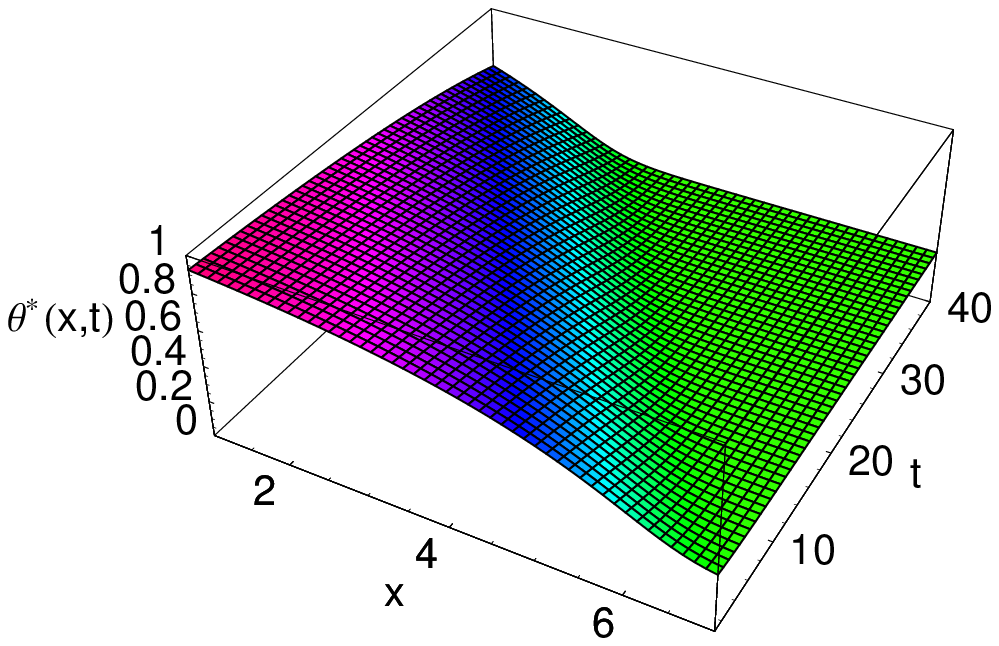}

\end{center}
\caption{
Graphs of $\phi(x,t)$ (left) and the optimal response function $\theta^*(x,t)$ (right).}
\label{fig:phithetaAdverse} 
\end{figure} 

For parameter values: $m=3/2, \alpha=\beta=0, c=-0.08, K_0=0.1$, and $\omega=1$ the function has three roots. An unstable root is located at $v^-\approx 0.888 <1$ and a stable root exists at $v^+\approx 4.488>1$. Therefore the traveling wave solution $\varphi(x,t)$ becomes strictly less than $1$ for $x\to -\infty$. Hence the optimal response function $\theta^*(x,t) = \min(1, \varphi^{-\frac{1}{m-1}}(x,t) )$ is identically 1 for all  sufficiently large negative $x$  (see Fig.~\ref{fig:phithetaAdverse2}).

\begin{figure}
\begin{center}
\includegraphics[width=0.45\textwidth]{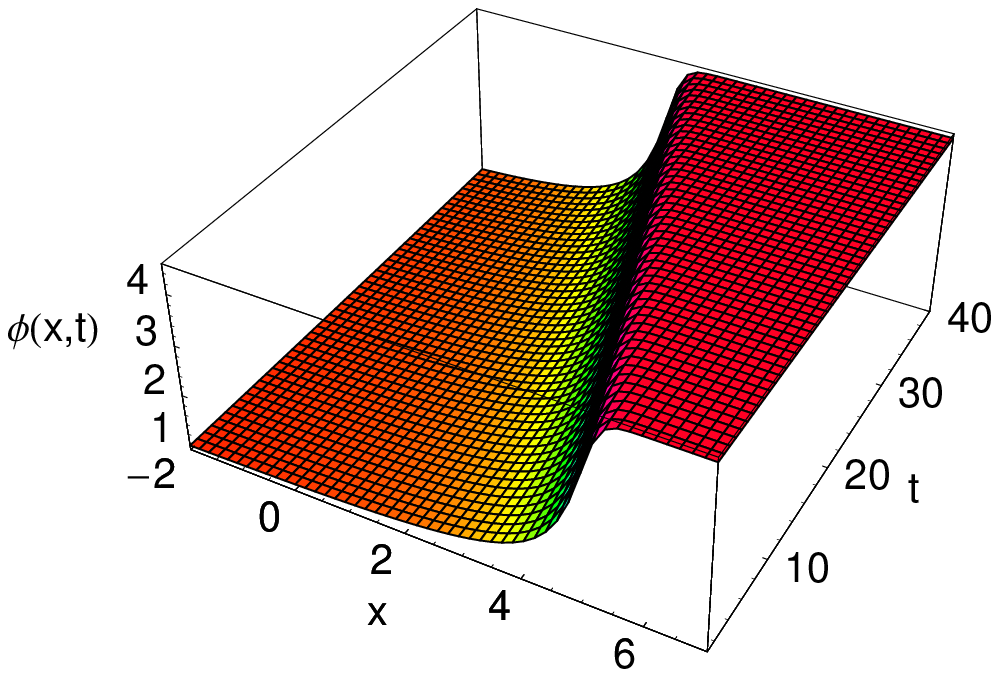}
\includegraphics[width=0.45\textwidth]{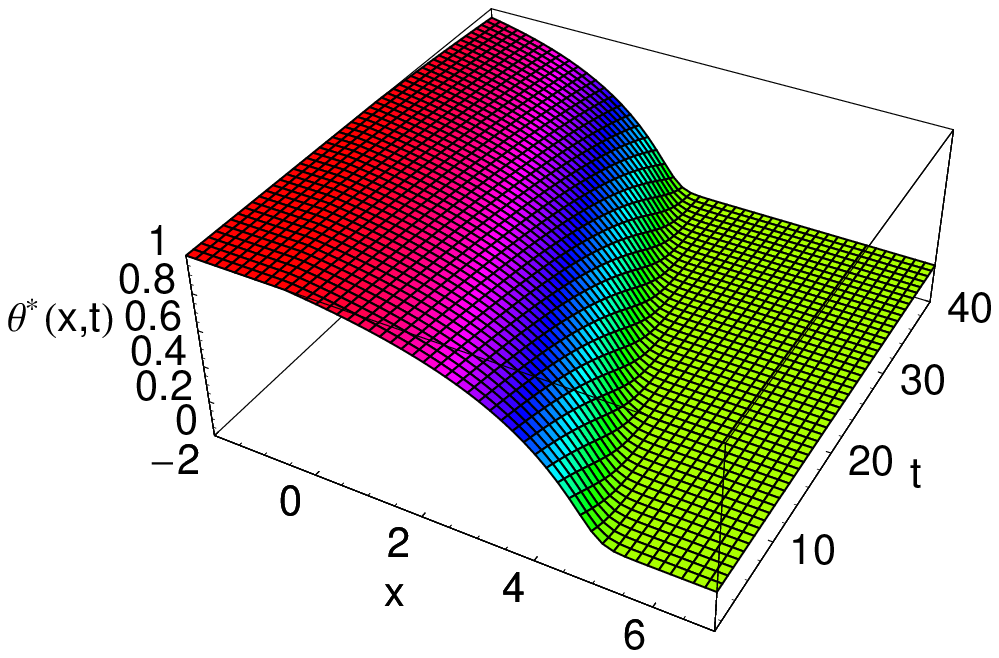}
\end{center}
\caption{
A graph of the optimal response function $\theta^*(x,t)$ attaining the prescribed 
boundary $\theta=1$ (left), a graph of the optimal value $\theta^*(x,t)$ computed by the stochastic dynamic optimization model \cite{KILIANOVA2006} (right). }
\label{fig:phithetaAdverse2} 
\end{figure} 

\medskip
Finally, we give a justification of the assumption made on the value function $V$. In what follows, we shall prove that $\partial^2_x V<0$ and $\partial_x V>0$. 

\begin{proposition}\label{prop:concavity}
Suppose that the terminal condition $V(x,T)\equiv u(x)$ is a smooth function and there exist constants $\lambda^\pm>0$ such that $\lambda^-< - u''(x)/u'(x) <\lambda^+$ for all $x\in \R$. Then $\lambda^-/\omega < \varphi(x,t) <\lambda^+/\omega$ for all $x\in \R$ and $t\in[0,T]$ where $\varphi$ is a solution to equation (\ref{HJB}) or (\ref{HJBX}) and satisfying the terminal condition $\varphi(x,T) = -(1/\omega) u''(x)/u'(x),\,  x\in\R$. 

Furthermore, for the value function $V$ it holds: $\partial^2_x V(x,t)<0$ and $\partial_x V(x,t)>0$ for all $x\in \R$ and $t\in[0,T]$
\end{proposition}

\begin{proof}
Equations (\ref{HJB}) as well (\ref{HJBX}) are strictly parabolic partial differential equations with a strictly positive diffusion coefficient $A^\prime(\phi)$ which is uniformly bounded from below and above. Using the parabolic maximum principle we conclude $\lambda^-/\omega < \varphi(x,t) <\lambda^+/\omega$ for all $x\in \R$ and $t\in[0,T]$, provided that $\lambda^-/\omega < \varphi(x,T) \equiv - (1/\omega) u''(x)/u'(x)  <\lambda^+/\omega$ for all $x\in \R$. 

Concavity and the monotonicity of the value function $V$ can be now deduced from the fact $-\partial^2_x V(x,t)/\partial_x V(x,t) = \omega \phi(x,t) \in (\lambda^-, \lambda^+)$.
\end{proof}

\section{A motivation for analyzing HJB equations with range bounds}

The optimal allocation problem has a long history of research and much progress has been made so far (see for instance Dupa\v{c}ov\'a \cite{D}). For a comprehensive overview of the stochastic dynamic optimization  problems of the form (\ref{maxproblem}) with the prescribed terminal value we refer the reader to papers by Merton \cite{M}, Browne \cite{BROWNE1995}, Bodie \emph{et al.} \cite{BODIE2003} and to the wide range of literature referenced therein. 

As an example of a process $X_t^\theta$ of the form (\ref{genproc}) one can consider a stochastic process representing returns on the accumulated sum of a saver's portfolio consisting of volatile stocks and less volatile bonds. The time discrete version of the  stochastic dynamic accumulation model has been proposed and analyzed by Kilianov\'a \emph{et al.}  in \cite{KILIANOVA2006}. In the time continuous limit, according to the pension savings accumulation model analyzed by Macov\'a and \v{S}ev\v{c}ovi\v{c} in \cite{MS}, we have $X_t=\ln Y_t$ where the stochastic variable $Y_t$ represents the logarithm of a ratio $Y_t$ of the accumulated sum in the pension account and the yearly salary of a future pensioner. 
 More precisely, by choosing a portfolio consisting of $\theta\in[0,1]$ part of stocks and $1-\theta$ part of bonds one can derive the SDE for the ratio $Y_t$:
\begin{equation}
\d Y_t^\theta = (\varepsilon + \mu^\theta Y_t) \d t + \sigma^\theta Y_t \d W_t,
\label{processY}
\end{equation}
where $\mu^\theta = \theta \mu_{s} + (1-\theta) \mu_{b}$ and
$(\sigma^\theta)^2 = \theta^2 (\sigma_{s})^2 + (1-\theta)^2 (\sigma_{b})^2 + 
2\varrho \theta(1-\theta) \sigma_{s}\sigma_{b}$ (c.f. \cite{MS}). In this model, $\mu_{s}, \mu_{b}$ and $\sigma_{s}, \sigma_{b}$  denote the expected returns and volatilities on stocks and bonds, respectively. Here $\varrho\in (-1,1)$ is a correlation between returns on stocks and bonds. Using It\=o's lemma for the variable $X_t=\ln Y_t$,  we obtain the following SDE:
\begin{equation}
\d X_t^\theta = \left( \varepsilon e^{-X_t} + \mu^\theta -\frac12 (\sigma^\theta)^2 \right) \d t + \sigma^\theta \d W_t.
\label{processX}
\end{equation}
In a stylized financial market, one can assume the expected return $\mu_b$ on bonds to be very small when compared to $\mu_s$ and also $\sigma_b \ll \sigma_s$. Furthermore, we notice that the yearly contribution rate $\varepsilon>0$ can be considered as a  small model parameter (the value $\varepsilon=0.09$ was accepted in Slovakia, $\varepsilon=0.03$ was proposed in Czech republic, and, $\varepsilon=0.14$ was accepted in Bulgarian pension fund system). The asymptotic expansions of a solution to the corresponding HJB equation with respect to the small parameter $0<\varepsilon\ll 1$ have been analyzed in \cite{MS}. By taking the formal asymptotic limit $\varepsilon\to 0, \mu_b\to 0, \sigma_b\to 0$ and rescaling the time $t\mapsto t/\sigma_s^2$, we end up with the following reduced SDE governing the variable $X_t$:
\begin{equation}
\d X_t^\theta = (\omega \theta - \frac12 \theta^2 ) \d t + \theta \d W_t, 
\label{sdeXlimit}
\end{equation}
where $\omega=\mu_s/\sigma_s^2$ is the Sharpe (reward-to-variability) ratio of returns on bonds. Such an underlying stochastic process corresponds to the one analyzed in Theorem~\ref{theorem:riskseeking}. The wave speed of the traveling wave solution from Theorem~\ref{theorem:riskseeking} is positive,  $c>0$ and the profile $v$ is a decreasing function, $v^\prime<0$. As a consequence, we have 
$\partial_t\varphi >0$ and $\partial_x\varphi < 0$. Hence, for the optimal response function $\theta^* = \min(1, 1/\varphi)$, we obtain the following properties:
\[
\partial_t \theta^*(x,t) \le 0, \qquad \partial_x \theta^*(x,t) \ge 0, \quad x\in\R,\ t\in(0,T).
\]
From the optimal allocation portfolio problem point of view, the decreasing behavior of the proportion of stocks $\theta^*$ in the portfolio with respect to time can be interpreted as the necessity to be more conservative in investing when the time horizon $T$ is approached. On the other hand, increasing dependence of $\theta^*$ with respect to the accumulated property $x$ indicates an increase of the investor's risk selection preference with the increase of the accumulated property $x$.

\section{Conclusions}
In this paper we have investigated traveling wave solutions to a fully nonlinear parabolic PDE arising from the model of a stochastic dynamic optimal allocation problem. Our approach is  based on the method solving and analyzing the fully nonlinear parabolic partial differential equation which can be constructed from the corresponding Hamilton--Jacobi--Bellman equation. Under certain assumptions made on the form of the underlying stochastic process, we found a range of parameters for which a traveling wave solution can be  constructed. More precisely, we identified parameters for which the traveling wave profile is monotonically decreasing or increasing and it has positive or negative wave speed, respectively.

\end{document}